\documentclass[11pt,letterpaper]{article}
\usepackage[letterpaper,margin=0.9in]{geometry}
\usepackage[T1]{fontenc}
\usepackage{lmodern,mathrsfs}
\usepackage{amsmath,amssymb,amsthm,mathtools,bm}
\usepackage{graphicx,booktabs,tabularx,array,enumitem}
\usepackage{microtype,xcolor}
\usepackage{tikz,needspace}
\usepackage[normalem]{ulem}
\usetikzlibrary{arrows.meta,positioning,calc}
\definecolor{figink}{HTML}{18364A}
\definecolor{figblue}{HTML}{267CA5}
\definecolor{figlight}{HTML}{E8F2F7}
\definecolor{figmid}{HTML}{BDD8E5}
\definecolor{figorange}{HTML}{D76B32}
\definecolor{figgrey}{HTML}{69777F}
\usepackage[numbers,sort&compress]{natbib}
\usepackage[colorlinks=true,linkcolor=blue!45!black,citecolor=blue!45!black,urlcolor=blue!45!black]{hyperref}
\hypersetup{pdftitle={Single-Shot Error Correction at Optimal Spacetime Cost},pdfauthor={Kishor Bharti, Tobias Haug, Derek Khu, Andrew Tanggara}}
\newtheorem{theorem}{Theorem}[section]
\newtheorem{lemma}[theorem]{Lemma}

\newtheorem{corollary}[theorem]{Corollary}
\newtheorem{fact}[theorem]{Fact}
\newtheorem{assumption}[theorem]{Assumption}
\theoremstyle{definition}

\theoremstyle{remark}

\definecolor{revisionblue}{RGB}{0,70,180}
\definecolor{revisionblue}{RGB}{0,0,0}
\newcommand{\rev}[1]{{\color{revisionblue}#1}}
\definecolor{darkred}{RGB}{120,0,0}
\definecolor{darkpurple}{RGB}{120,0,120}
\definecolor{darkred}{RGB}{0,0,0}
\definecolor{darkpurple}{RGB}{0,0,120}

\newcommand{\atchg}[1]{\textcolor{darkred}{#1}}

\setlist{itemsep=3pt,topsep=5pt,leftmargin=1.6em}
\numberwithin{equation}{section}
\title{\LARGE\bfseries Single-Shot Error Correction at Optimal Spacetime Cost}
\author{Kishor Bharti$^{1,2,\star}$, Tobias Haug$^3$, Derek Khu$^4$, and Andrew Tanggara$^{5,6}$\\[5pt]
\small $^1$QuICS, NIST/University of Maryland, College Park, Maryland 20742, USA\\
\small $^2$UMIACS, University of Maryland, College Park, Maryland 20742, USA\\
\small $^\star$ Currently at IonQ Inc.\\
\small $^3$Quantum Research Center, Technology Innovation Institute, Abu Dhabi, United Arab Emirates\\
\small $^4$Institute of Advanced Intelligence and Computing, A*STAR, Singapore\\
\small $^5$Centre for Quantum Technologies, National University of Singapore, Singapore 117543\\
\small $^6$Nanyang Quantum Hub, Nanyang Technological University, Singapore 639673}
\date{}
\begin{document}
\maketitle
\begin{abstract}
Recent work established that, for independent erasures, storing $K$ logical qubits for $S$ time steps with error at most $\varepsilon$ requires a spacetime cost of $\Omega(S(K+\log(S/\varepsilon)))$. A matching construction was also given, but assumes ideal error correction. Here we show that the same scaling can be achieved with an explicit noisy error-correction circuit and efficient decoding, counting all state-preparation, gate, measurement, and wait locations, provided that the hardware supports long-range qubit connectivity and fast, reliable classical processing. Thus, reliability adds only a logarithmic overhead shared by all stored qubits. The construction works for sufficiently weak but otherwise general circuit noise, including faults during syndrome extraction and correlated faults, and for independent erasures it matches the best known lower bound up to constant factors. We attain this scaling using quantum Tanner codes, with each memory step consisting of one round of syndrome extraction followed by a fixed number of parallel classical decoder steps that reduce the residual error enough to keep later faults correctable, regardless of memory size, storage time, or target accuracy. Together, this gives exponentially small failure probability per step while keeping the cost of each memory step linear in the code size. The same analysis also extends to certain constant-depth logical Clifford operations.

\end{abstract}

\section{Introduction and main results}
\label{sec:introduction}
A quantum memory must preserve information while the operations used to protect it are themselves noisy. The threshold theorem shows that sufficiently weak noise can be overcome~\cite{Shor1996,AharonovBenOr1997,AliferisGottesmanPreskill2006}. The remaining question is how many physical resources this protection requires. Error correction consumes qubits, gates, measurements, and time, all of which contribute to the cumulative spacetime cost of storing a register.

For independent erasures, recent work established the lower bound
\begin{equation}
 C=\Omega\!\left(S\left(K+\log\frac{S}{\varepsilon}\right)\right)
 \label{eq:prior-bound}
\end{equation}
for storing $K$ logical qubits for $S$ time steps with error at most $\varepsilon$, and gave a matching construction under idealized error correction~\cite{BhartiHaugTanggara2026}. However, attaining this bound with a realistic noise model and noisy recovery requires more than a good code. Faulty syndrome measurements can lead to incorrect corrections, and the data continue to accumulate errors while the classical decoder runs.

Here we show that an explicit noisy recovery circuit with efficient decoding attains the same scaling. We use quantum Tanner codes and count every preparation, gate, measurement, and wait within the memory protocol. \rev{The construction tolerates coherent errors and correlations between faults, provided that the faulty locations satisfy the local-stochastic condition specified below.} We assume that the register enters and leaves the protocol encoded; a final destructive logical measurement is also possible. We further show that the same recovery method protects a class of logical Clifford operations at the same asymptotic cost per layer. 

Figure~\ref{fig:memory-concept} summarizes our method. Each cycle measures the checks once and runs a fixed number of parallel decoder iterations. Recovery need not eliminate every error: it must reduce the accumulated error enough for subsequent faults to remain correctable. This allows one code block to provide both capacity and reliability, with a logarithmic contribution to the block size shared by all stored qubits.

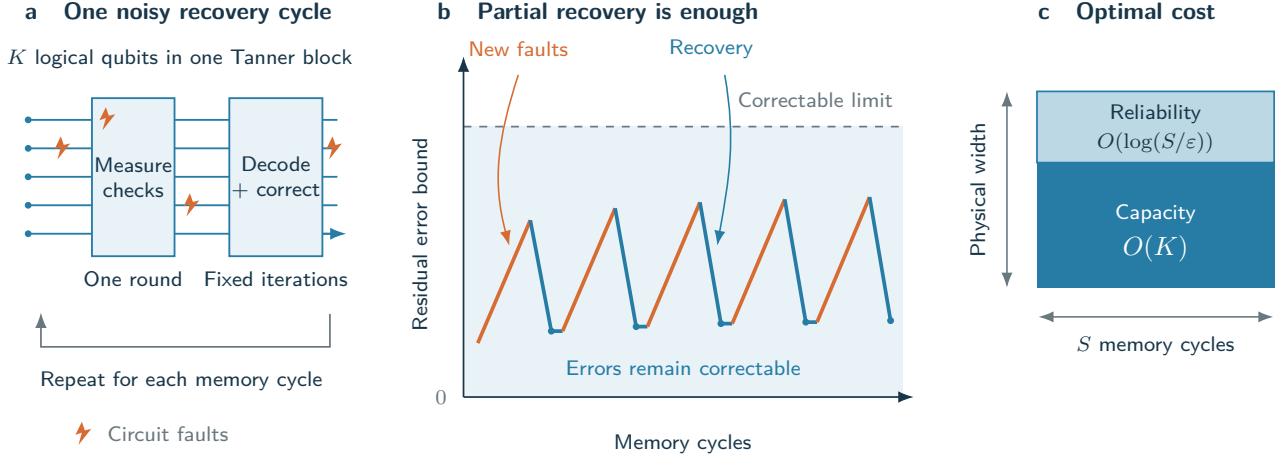
\begin{figure}[t]
\centering
\resizebox{\textwidth}{!}{%
\begin{tikzpicture}[x=1cm,y=1cm,>=Latex,
  every node/.style={font=\sffamily\fontsize{8}{9.5}\selectfont,text=figink},
  head/.style={font=\sffamily\bfseries\fontsize{9}{11}\selectfont,anchor=west},
  fault/.pic={\fill[figorange] (-.015,.16)--(-.10,-.015)--(-.015,.015)--(-.06,-.15)--(.11,.065)--(.025,.045)--cycle;}]
\path[use as bounding box] (-.05,-.20) rectangle (17.1,6.4);
\node[head] at (0,6.05) {\textbf{a}\quad One noisy recovery cycle};
\node at (2.20,5.43) {$K$ logical qubits in one Tanner block};
\foreach \y in {3.10,3.48,3.86,4.24,4.62} {
  \draw[figblue,line width=.7pt] (.18,\y)--(4.30,\y);
  \fill[figblue] (.18,\y) circle (.045);
}
\draw[->,figblue,line width=.7pt] (4.05,3.10)--(4.42,3.10);
\filldraw[fill=figlight,draw=figblue,line width=.7pt] (1.03,2.80) rectangle (2.13,4.91);
\filldraw[fill=figlight,draw=figblue,line width=.7pt] (2.86,2.80) rectangle (4.10,4.91);
\node[align=center] at (1.58,3.85) {Measure\\checks};
\node[align=center] at (3.48,3.85) {Decode\\+ correct};
\pic at (.63,4.24) {fault};
\pic at (1.23,4.63) {fault};
\pic at (2.35,3.48) {fault};
\pic at (4.25,4.24) {fault};
\node at (1.58,2.50) {One round};
\node at (3.48,2.50) {Fixed iterations};
\draw[->,figgrey,line width=.65pt] (4.20,2.03)--(4.20,1.59)--(.35,1.59)--(.35,2.05);
\node at (2.22,1.14) {Repeat for each memory cycle};
\pic at (.91,.43) {fault};
\node[anchor=west,text=figgrey] at (1.10,.43) {Circuit faults};
\begin{scope}[xshift=5.5cm]
\node[head] at (0,6.05) {\textbf{b}\quad Partial recovery is enough};
\fill[figlight] (.48,.92) rectangle (6.34,4.53);
\draw[figgrey,dashed,line width=.65pt] (.48,4.53)--(6.34,4.53);
\node[anchor=south east,text=figgrey] at (6.34,4.66) {Correctable limit};
\draw[->,line width=.65pt,figink] (.48,.92)--(.48,5.47);
\draw[->,line width=.65pt,figink] (.48,.92)--(6.48,.92);
\node[rotate=90] at (-.10,3.10) {Residual error bound};
\node[anchor=east,text=figgrey] at (.39,.92) {$0$};
\node at (3.40,.30) {Memory cycles};
\draw[figorange,line width=1.4pt] (.67,1.64)--(1.37,3.28);
\draw[figblue,line width=1.4pt] (1.37,3.28)--(1.65,1.80)--(1.80,1.80);
\draw[figorange,line width=1.4pt] (1.80,1.80)--(2.50,3.44);
\draw[figblue,line width=1.4pt] (2.50,3.44)--(2.78,1.86)--(2.93,1.86);
\draw[figorange,line width=1.4pt] (2.93,1.86)--(3.63,3.52);
\draw[figblue,line width=1.4pt] (3.63,3.52)--(3.91,1.90)--(4.06,1.90);
\draw[figorange,line width=1.4pt] (4.06,1.90)--(4.76,3.56);
\draw[figblue,line width=1.4pt] (4.76,3.56)--(5.04,1.92)--(5.19,1.92);
\draw[figorange,line width=1.4pt] (5.19,1.92)--(5.89,3.59);
\draw[figblue,line width=1.4pt] (5.89,3.59)--(6.17,1.94);
\foreach \x/\y in {1.65/1.80,2.78/1.86,3.91/1.90,5.04/1.92,6.17/1.94}
  \fill[figblue] (\x,\y) circle (.046);
\node[text=figorange] at (1.21,5.57) {New faults};
\draw[->,figorange,line width=.6pt] (1.21,5.22) to[out=-108,in=105] (1.09,2.94);
\node[text=figblue] at (3.78,5.57) {Recovery};
\draw[->,figblue,line width=.6pt] (3.96,5.22) to[out=-80,in=80] (3.86,2.94);
\node[text=figblue] at (3.41,1.31) {Errors remain correctable};
\end{scope}
\begin{scope}[xshift=13.50cm]
\node[head] at (0,6.05) {\textbf{c}\quad Optimal cost};
\draw[<->,figgrey,line width=.6pt] (-.28,2.38)--(-.28,5.00);
\node[rotate=90] at (-.64,3.69) {Physical width};
\fill[figblue] (.12,2.38) rectangle (3.28,4.05);
\filldraw[fill=figmid,draw=figblue,line width=.65pt] (.12,4.05) rectangle (3.28,5.00);
\node at (1.70,4.70) {Reliability};
\node at (1.70,4.32) {$O(\log(S/\varepsilon))$};
\node[text=white] at (1.70,3.37) {Capacity};
\node[text=white,font=\sffamily\fontsize{10}{11}\selectfont] at (1.70,2.91) {$O(K)$};
\draw[<->,figgrey,line width=.6pt] (.14,1.99)--(3.27,1.99);
\node at (1.70,1.60) {$S$ memory cycles};
\end{scope}
\end{tikzpicture}}
\caption{\textbf{Optimal quantum memory through repeated partial recovery.}
\textbf{a)} Each cycle measures every stabilizer check once, runs a fixed number of parallel decoder iterations, and applies a Pauli correction or updates the Pauli frame. Orange symbols indicate faults in the quantum circuit; we assume that classical processing is reliable.
\textbf{b)} New faults increase the error bound, while recovery reduces it without removing every error. The curve is schematic and assumes that every cycle stays within the allowed fault count. Exceeding that count in one cycle is exponentially unlikely in the block length.
\textbf{c)} The bar's height represents the physical width and its horizontal extent represents the storage duration. The two shades show the capacity and reliability contributions to the required width. Each cycle has linear physical cost, giving $O(S(K+\log(S/\varepsilon)))$ total locations.
}
\label{fig:memory-concept}
\end{figure}

\subsection{Optimal memory cost}
A \emph{memory cycle} is one complete error-correction step: one measurement round of each stabilizer-check type, a fixed number of parallel decoder iterations, and a Pauli correction or an update of the Pauli frame. \rev{Under the classical-control and connectivity assumptions below, one cycle takes a fixed number $\tau_{\rm cyc}$ of elementary time steps, including the waits during decoding. Once the code family, decoder iteration count, and hardware timings are fixed, $\tau_{\rm cyc}$ is independent of $N$, $K$, $S$, and $\varepsilon$. Storing the register for $S$ cycles therefore takes $S\tau_{\rm cyc}$ elementary time steps.}

\rev{At a faulty circuit location, we allow an arbitrary quantum operation on the qubits involved and a shared environment. This includes coherent errors, such as an unintended unitary rotation, whose Pauli components can interfere, as well as correlated errors where faults at different locations need not be independent. Dependence of quantum operations on earlier measurement outcomes is also allowed.}
The random set of those locations is independent of the logical input and satisfies a local-stochastic bound of strength $p$. Classical processing is reliable, and one parallel decoder iteration takes a fixed physical time. The required check interactions are available without a routing delay that grows with the code length. Section~\ref{sec:model} specifies these assumptions and the logical diamond-norm error criterion.

\Needspace{14\baselineskip}
\begin{theorem}[Optimal encoded quantum memory]
\label{m:memory}
There is a family of quantum Tanner codes and a threshold $p_*>0$ such that, at every fixed $p<p_*$, an encoded register of $K$ logical qubits can be stored for $S$ memory cycles with logical diamond-norm error at most $\varepsilon$. The construction uses a code block length $N$, peak quantum width $Q_{\rm peak}$, and total number of physical locations $C$ satisfying
\begin{equation}
 N=O\!\left(K+\log\frac S\varepsilon\right),\qquad
 Q_{\rm peak}=O(N),\qquad C=O(SN).
 \label{eq:memory-result}
\end{equation}
Each cycle measures every listed stabilizer check once and uses a fixed number of parallel single-shot decoder iterations. Its quantum depth, including decoder waits and an optional physical Pauli correction, is constant. Classical decoding requires $O(N)$ work per cycle. One fixed ideal recovery and decoding map defines the logical error for all \rev{fault paths} and measurement records.
\end{theorem}

The upper bound applies throughout the stated noise model. Independent erasures are one allowed case, and there the bound matches Eq.~\eqref{eq:prior-bound} up to constants; Corollary~\ref{m:tight} makes this comparison precise. Thus allowing noisy recovery does not increase the optimal memory scaling.

The resulting relative spacetime overhead is
\begin{equation}
 \frac{C}{KS}=O\!\left(1+\frac{\log(S/\varepsilon)}{K}\right).
 \label{eq:relative-overhead}
\end{equation}
For a sufficiently wide register this overhead is constant. For a small register, longer storage or higher accuracy requires additional redundancy. %

\subsection{Why constant-time recovery suffices}
The construction uses the single-shot decoder of Gu et al.~\cite{GuEtAl2024}, based on the decoding algorithms of Leverrier and Z\'emor~\cite{LeverrierZemor2023}. Theorem~3.5 and Corollary~3.6 of Ref.~\cite{GuEtAl2024} already establish repeated recovery and its accumulated failure bound. Here we combine these recovery guarantees with the circuit and fault model below to derive the cumulative memory cost, including the quantum waits during decoding.

The proof has three steps. First, bounded check weight and incidence give a syndrome-extraction circuit with constant depth and $O(N)$ locations. A fault can therefore affect only a bounded number of data qubits and syndrome bits. Second, the decoder contracts earlier errors sufficiently to accommodate the new faults in each cycle. A cycle with at most a sufficiently small fixed fraction of faulty locations leaves a correctable residual error. Third, the local-stochastic bound makes a cycle exceeding this fraction exponentially unlikely in $N$. A union bound over $S$ cycles then gives the logarithmic block-size requirement in Theorem~\ref{m:memory}.

For coherent faults, a bound on individual Pauli errors is not enough on its own. We show that all error operators arising on successful \rev{fault paths} belong to one correctable linear space. The Knill--Laflamme conditions then give a common recovery, preserving interference between Pauli terms and entanglement with an external reference. This converts the circuit analysis into a logical diamond-norm bound without assuming a stochastic Pauli error after each cycle.

\subsection{Logical Clifford operations}
The memory need not remain idle between recovery cycles. We show that the same bound extends to a prescribed sequence of $T$ logical Clifford layers whenever each layer has a constant-depth physical implementation with bounded error spread and the output code retains the same decoder guarantees. For a fixed number of blocks, with these implementations and \rev{the allocation of logical qubits to code blocks} available uniformly along the code family, the total cost is
\begin{equation}
 C=O\!\left(T\left(K+\log\frac{T}{\varepsilon}\right)\right).
 \label{eq:intro-clifford-cost}
\end{equation}
The input and output remain encoded. Transversal CNOT between identical blocks is an explicit example: every corresponding logical pair undergoes a CNOT in one physical gate layer. Specified logical Pauli products are also supported, \rev{as is a Hadamard gate on every physical qubit, with the output interpreted in the CSS code obtained by exchanging the $X$-type and $Z$-type stabilizer generators.} The extra error spread from these gates is bounded, so a fixed number of decoder iterations suffices after each layer.

We also characterize the physical CNOT circuits that realize a prescribed logical CNOT matrix between two identical CSS blocks. The allowed physical matrices form an explicit affine space. A representative with bounded row and column weights gives a constant-depth circuit with $O(N)$ gates and therefore meets the recovery theorem's requirements. This identifies the remaining task for selecting individual logical pairs within a Tanner block: finding such sparse representatives uniformly and efficiently. The characterization and its sparsity criterion are proved in Appendix~\ref{s:cnots}; their existence for every selection of logical pairs is an additional assumption, not a consequence of the memory construction.

\subsection{Relation to previous work}
Good quantum LDPC codes combine constant encoding rate, linear distance, and sparse stabilizer checks~\cite{PanteleevKalachev2022,LeverrierZemor2022}. Efficient decoders make these parameters useful for fault tolerance~\cite{DinurEtAl2023,LeverrierZemor2023}. Single-shot error correction, introduced by Bomb\'in~\cite{Bombin2015}, uses one noisy measurement round rather than a growing number of repetitions. Campbell developed its adversarial-noise formulation~\cite{campbell2019theory}, and Gu et al. proved single-shot decoding and constant-iteration repeated recovery for good quantum Tanner codes~\cite{GuEtAl2024}. These results supply the code and decoder used here. We combine them with a complete circuit analysis, including coherent fault operations and decoder waits, to obtain the memory bound with explicit dependence on $K$, $S$, and $\varepsilon$.

Conventional code families illustrate why the encoding rate matters for this cost. Storing each logical qubit in a separate distance-$d$ surface-code patch uses $\Theta(Kd^2)$ physical qubits~\cite{DennisEtAl2002,FowlerEtAl2012}. With the usual below-threshold estimate of an error probability exponentially small in $d$ per syndrome cycle, taking $d$ of order $\log(KS/\varepsilon)$ gives a memory cost of order $KS\log^2(KS/\varepsilon)$. Here $S$ counts physical syndrome cycles. This comparison concerns independent patches with local interactions; our Tanner construction uses long-range check connectivity and shares the reliability contribution across the register. The different connectivity assumptions are essential, as geometric restrictions can impose growing overhead~\cite{BaspinFawziShayeghi2023}.

For general quantum computation, Gottesman's framework and its realization with quantum expander codes establish constant space overhead~\cite{Gottesman2014,FawziGrospellierLeverrier2018}. Their original universal constructions incur polynomial time overhead. Subsequent protocols reduce this to quasi-polylogarithmic overhead~\cite{YamasakiKoashi2024}, polylogarithmic overhead including classical processing~\cite{TamiyaKoashiYamasaki2024}, and logarithmic overhead up to sub-polylogarithmic factors~\cite{NguyenPattison2025}; Han and Liu report strictly logarithmic time overhead~\cite{HanLiu2026}. These asymptotic statements use the respective circuit-size parameters and computational models of those works. Constant space overhead has also been extended to general noise~\cite{ChristandlFawziGoswami2025}. Our constant duration per protected layer applies to encoded memory and the stated Clifford implementations. It does not replace the broader gate and interface constructions needed for universal computation. The memory bound also complements lower bounds on the physical width required for fault tolerance~\cite{FawziMullerHermesShayeghi2022}.

\rev{Constant-depth logical Clifford circuits also have precedents. Zheng et al. apply a constant number of teleportation or measurement steps to the encoded data, using encoded ancillas prepared beforehand~\cite{ZhengEtAl2020}. Their resource analysis includes an average $O(N^2/\log N)$ physical-gate cost for an arbitrary Clifford circuit on a constant-rate block. The depth of the circuit acting on the data must therefore be distinguished from the cost of preparing its ancillas.} Our $O(N)$ cost applies to the restricted physical layers specified above, together with their noisy recovery. Logical Pauli measurement schemes~\cite{CohenEtAl2022,HuangEtAl2023} and high-rate surgery~\cite{ZhengJiangXu2025} provide other approaches to logical operations. Golowich, Chang, and Zhu construct constant-quantum-spacetime-overhead addressable gates by \rev{single-shot code switching} between product codes, allowing polynomial classical computation~\cite{GolowichChangZhu2025}.
Tan et al. achieve single-shot state preparation, error correction, and universal logical gates through constant-depth quantum gadgets that switch between two- and three-dimensional hypergraph-product codes~\cite{tan2025singleshotuniversalityquantumldpc}. 
Their construction uses a different code family and does not supply the sparse Tanner representatives assumed here. %

Section~\ref{sec:model} specifies the noise model and coding ingredients. Section~\ref{sec:memory} proves the memory theorem and its readout extension. Section~\ref{sec:clifford} states the Clifford extension and explains its scope. Appendix~\ref{app:clifford-proof} gives its proof, and Appendix~\ref{s:cnots} derives the CNOT characterization and the conditional cost for selected logical pairs.

\section{Noise model and coding ingredients}
\label{sec:model}
The resource bound depends on both the physical noise and the time required for classical recovery. We make these assumptions explicit before stating the Tanner-code and decoder guarantees used in the proof. Throughout, $K,S,T\ge1$, $0<\varepsilon\le1/2$, and logarithms are natural. The parameter $N$ denotes the number of data qubits in a code block; $S$ counts memory cycles and $T$ counts logical layers in the computation extension.

\subsection{Circuit locations and classical control}
We use the standard distinction between a \emph{fault}, which occurs at a circuit location, and an \emph{error}, which is its effect on the quantum state~\cite{GottesmanBook2024}. A location is a preparation, a gate acting on a fixed maximum number of qubits, a measurement, or a one-qubit wait for one elementary time step. Resets can be implemented by discarding and preparing a qubit. A two-qubit gate counts as one location; counting both operands changes the bounds only by a constant. The total cost $C$ counts all locations, including waits on every \atchg{data qubit and ancilla while it remains in the circuit}. 
The peak number of physical qubits present in the circuit is denoted by $Q_{\rm peak}$.
\atchg{A qubit is counted from its preparation or entry into the protocol until it is destructively measured or discarded.}

The hardware supports the specified check interactions and any additional gate interactions required by a theorem, without routing depth growing with $N$. Classical processing and stored measurement records are reliable. For the Tanner decoder, a fixed number of local parallel iterations takes fixed physical time. All quantum waits during this processing contribute to $C$ and are subject to noise. The $O(N)$ classical work per cycle is stated separately from the quantum-location count. Constant latency assumes sufficient parallel classical hardware and fixed-time communication along the decoder graph; it is not a claim about a fixed-size classical processor.

\subsection{Faults and logical error}
\rev{A \emph{fault path} specifies the set $F$ of faulty circuit locations~\cite{GottesmanBook2024}. We sample this set from a classical distribution independently of the logical input.} For every prescribed set $A$ of circuit locations,
\begin{equation}
 \Pr(A\subseteq F)\le p^{|A|}.
 \label{eq:ls}
\end{equation}
Thus any specified set of $r$ locations is simultaneously faulty with probability at most $p^r$. This local-stochastic condition allows correlations without requiring independent faults. At a faulty location, an arbitrary completely positive trace-preserving (CPTP) operation may act on the location's operands and a shared environment. Faulty preparations, resets, and measurements are local instruments; a measurement fault may also change its reported bit. The environment is initially independent of the logical input and its external reference. Fault operations may depend causally on earlier records and may be correlated across locations. An ideal location performs its prescribed operation without an additional coupling to the environment.

\rev{The model includes coherent errors for a fixed fault path, local-stochastic Pauli faults, and erasure faults whose locations are known and whose qubits are replaced. For example, an unintended rotation $e^{-i\vartheta P}=\cos\vartheta\,I-i\sin\vartheta\,P$ about a \atchg{Hermitian Pauli operator $P$, with $P^2=I$ and $\vartheta\in\mathbb R$,} acts as a coherent sum of Pauli terms, rather than a probabilistic choice between $I$ and $P$.} It does not allow a coherent superposition of distinct \rev{classical fault paths}. In particular, a small coherent overrotation at every location is not represented by calling those locations fault-free. Persistent leakage requires a separate analysis unless its influence remains within the bounded sets of locations and qubits used in the proof.

Let $\mathrm{Enc}$ be a fixed ideal encoding isometry and $\mathcal R$ a fixed ideal recovery followed by decoding. For a noisy physical memory channel $\mathcal P$, we require
\begin{equation}
 \|\mathcal R\circ\mathcal P\circ\mathrm{Enc}-\mathrm{id}_K\|_\diamond
 \le\varepsilon.
 \label{eq:error}
\end{equation}
The diamond norm includes arbitrary entanglement with an external reference. The ideal maps at the input and output define the error criterion; they are not physical circuits included in $C$. The input is therefore already encoded. The output may contain a correctable residual error and need not be close to the codespace in physical trace distance.

A Pauli correction may be applied physically or recorded in a \emph{Pauli frame}, a classical record of the correction to be accounted for in later operations. In the latter case the output includes the reliable frame register $f$. If $P_f$ is the recorded Pauli, the fixed comparison map first applies
\begin{equation}
 \sum_f |f\rangle\langle f|\otimes\rho_f
 \longmapsto\sum_f P_f\rho_fP_f^\dagger,
\end{equation}
and then performs recovery and decoding. This defines one channel, independent of the \rev{fault path}. For destructive readout we instead compare the decoded quantum-to-classical channels, including the conditional reference states. Unused syndrome and auxiliary records are discarded.

\subsection{Quantum Tanner codes and single-shot decoding}
A binary CSS code is specified by check matrices $H_X,H_Z$ with $H_XH_Z^T=0$. Their row spaces $\mathcal S_X$ and $\mathcal S_Z$ specify the $X$ and $Z$ stabilizers. Check lists may contain redundant rows. We write $k_N$ for the number of logical qubits and $d_N$ for the code distance. Paired logical representatives $x_i\in\ker H_Z$ and $z_i\in\ker H_X$ satisfy $x_i\cdot z_j=\delta_{ij}$. We fix the encoding
\begin{equation}
 \mathrm{Enc}_Q|a\rangle
 =2^{-\dim(\mathcal S_X)/2}
   \sum_{u\in\mathcal S_X}|u+\textstyle\sum_i a_ix_i\rangle.
 \label{eq:css-encoding}
\end{equation}
\rev{The chosen $X$-type and $Z$-type stabilizer generators have phase $+1$, and the codespace is their simultaneous $+1$ eigenspace.} This choice fixes the logical coordinates and phases used in the gate and readout statements; changing the logical basis during a computation is an additional operation. For memory with $K<k_N$, the encoding places the input in $K$ fixed logical coordinates and initializes the unused coordinates to $|0\rangle$. The ideal output map decodes the full block and discards those unused coordinates. This initialization is part of the encoded-input assumption.

For a binary error vector $e$, its \emph{reduced weight} $|e|_R$ is the minimum weight among vectors that differ from $e$ by a stabilizer of the same Pauli type. Such vectors act identically on the code. Single-shot decoding bounds the reduced weight left by one noisy syndrome round, rather than requiring that every residual error be removed.
\begin{fact}[Tanner codes and single-shot decoding]
\label{m:decoder}
The quantum Tanner families used in Refs.~\cite{LeverrierZemor2022,GuEtAl2024} have constants $R,\delta,w,q>0$ such that
\begin{equation}
 k_N\ge RN,\qquad d_N\ge\delta N,
 \label{eq:tanner}
\end{equation}
and every measured check has weight at most $w$, while each qubit belongs to at most $q$ checks. For either CSS error component, a deterministic decoder $D_j$, run for $j$ parallel iterations, has constants $a,b,c,\beta>0$ and $\alpha_j\to0$ exponentially in $j$ such that
\begin{align}
 a|e|_R+b|m|&\le cN
 \quad\Longrightarrow\quad
 |e+D_j(\operatorname{syn}(e)+m)|_R
 \le\alpha_j|e|_R+\beta|m|.
 \label{eq:decoder}
\end{align}
Here $m$ is the error in the reported syndrome, and $\operatorname{syn}(e)$ is the ideal syndrome. Constants absorb the fixed conversion between erroneous checks and erroneous local syndrome records. Fixed $j$ requires $O(N)$ classical work and a fixed number of local communication steps.
\end{fact}

The decoder acts separately on the $X$ and $Z$ error components and receives a classical syndrome. The data entering recovery need not lie in the codespace: the decoder theorem bounds the Pauli error relative to an encoded state. The reduced-weight form of Eq.~\eqref{eq:decoder} follows by applying the decoder theorem to a minimum-weight stabilizer representative: this replacement changes neither the syndrome nor the decoder output. The relevant statements are Theorems 3.4 and 4.20 of Ref.~\cite{GuEtAl2024}. We use their deterministic error bound. Section~\ref{sec:memory} extends the circuit analysis to coherent errors, without requiring a local-stochastic distribution for the residual error.

To choose a code for every target register size and accuracy, we also need available block lengths within a constant factor of any sufficiently large target. Namely, there are $N_0$ and $c_N$ such that each $x\ge N_0$ has an available $N\in[x,c_Nx]$. For the Tanner family this follows from the arithmetic tower in Claim 6.7 of Ref.~\cite{DinurEvraLivneLubotzkyMozes2022}. The construction uses $G_i=\mathrm{PSL}_2(q_*^i)$ for a fixed odd prime power $q_*$, with the required generating sets for all sufficiently large $i$. Since $N_i$ is a fixed multiple of $|G_i|$,
\begin{equation}
 \frac{N_{i+1}}{N_i}
 =q_*\frac{q_*^{2i+2}-1}{q_*^{2i}-1}
 \le q_*(q_*^2+1).
 \label{eq:tanner-density}
\end{equation}
Taking the first available length above $x$ gives the claim. This step is needed for the resource bound: good rate and distance alone would not control gaps between available lengths. All constants in what follows may depend on the fixed code family, physical noise strength, and circuit schedule.

\section{Optimal quantum memory}
\label{sec:memory}
The proof of Theorem~\ref{m:memory} follows the recovery cycle illustrated in Fig.~\ref{fig:memory-concept}a. We first construct the syndrome-extraction circuit and bound the effect of new faults. We then combine this bound with the \atchg{reduction of the incoming-error contribution by the factor $\alpha_j<1$ in Eq.~\eqref{eq:decoder}} to show that errors remain correctable over arbitrarily many cycles, except when a cycle contains too many faults.
The probability of that event determines the block size needed for the target accuracy.

\subsection{Syndrome extraction and recovery}
For an $X$ check, prepare one ancilla in $|+\rangle$, apply CNOT gates from the ancilla to the check's data qubits, and measure the ancilla in the $X$ basis. For a $Z$ check, prepare one ancilla in $|0\rangle$, apply CNOTs from the data qubits to the ancilla, and measure the ancilla in the $Z$ basis. Complete all $X$-check circuits before the $Z$-check circuits so that their ordering implements the intended syndrome measurement.

For either check type, the data--check incidence graph is bipartite and has maximum degree at most $\max(w,q)$. Edge coloring schedules its CNOTs in at most that many layers. Ancilla preparation and measurement add constant depth. Bounded incidence gives $O(N)$ nonzero checks and gates, so one syndrome round has $O(N)$ locations, including all waits. Because circuit depth and gate size are bounded, a single fault affects only a bounded number of data qubits and reported syndrome bits.

After the syndrome round, run $j$ parallel decoder iterations on each CSS component and apply the resulting Pauli correction. Equivalently, record the correction in a Pauli frame and interpret the next syndrome relative to that frame. In either implementation, the data wait for the decoder and these waits are charged. Conjugating a fault operation by a known Pauli does not change its support, so the same noise model applies when corrections are tracked in a frame.

\subsection{A recovery that also corrects coherent errors}
For a Pauli $P=X^xZ^z$ on one code block, define
\begin{equation}
 r(P)=\min_{u\in\mathcal S_X}|x+u|
      +\min_{v\in\mathcal S_Z}|z+v|.
 \label{eq:radius}
\end{equation}
The quantity $r(P)$ measures the error after allowing multiplication by $X$ and $Z$ stabilizers. It is subadditive and unchanged by multiplication by a stabilizer. A Pauli with $r(P)\le r$ has a stabilizer-equivalent representative of physical weight at most $r$. Let $\mathsf E(r)$ be the linear span of these Paulis. Working with the entire span allows us to keep coherent combinations of errors throughout the proof.

\begin{lemma}[Common recovery]
\label{m:common}
If $2r<d_N$, a single ideal recovery corrects $\mathsf E(r)$ on the entire codespace, including coherent combinations 
\atchg{of the form $\sum_\nu P_\nu\otimes B_\nu$, where each $P_\nu$ is a Pauli with $r(P_\nu)\le r$ and $B_\nu$ is an operator on the shared environment}.
\end{lemma}

\begin{proof}
Let $\Pi$ project onto the codespace. For Pauli generators $P,Q$ of this span, $P^\dagger Q$ has a stabilizer-equivalent representative of weight at most $2r<d_N$. If it has nonzero syndrome, $\Pi P^\dagger Q\Pi=0$. If its syndrome is zero, it is a stabilizer up to phase, so $\Pi P^\dagger Q\Pi=c_{PQ}\Pi$. Expanding arbitrary linear combinations of these Pauli errors gives the Knill--Laflamme conditions for the whole span~\cite{KnillLaflamme1997}. 
\rev{An ideal recovery for the error criterion in Eq.~\eqref{eq:error}} measures the ideal syndrome and inverts 
\atchg{one fixed Pauli $P_s$ with syndrome $s$ and $r(P_s)\le r$, for each syndrome $s$ attained by such a Pauli}. Any two 
\atchg{Paulis $P_s,Q_s$ with the same syndrome and $r(P_s),r(Q_s)\le r$} act identically on the code up to a scalar.
Extending arbitrarily on other syndromes gives a CPTP map. 
Linearity retains all cross terms.
\end{proof}

\subsection{Faults in a complete cycle}

The decoder receives the actual classical measurement record. For a coherent fault, different Pauli terms can be compatible with that same record, and all must receive the same correction. The following lemma bounds their data and syndrome errors without treating those terms as a classical mixture.

\begin{lemma}[\rev{Fault propagation and the observed syndrome}]
\label{m:record}
Fix the original \rev{fault path} in \rev{one complete error-correction cycle} and its observed syndrome record $y$. Let $f$ be its number of new faulty locations, and let the incoming error operators lie in $\mathsf E(r_{\rm in})$. 
\atchg{Let $\Pi_{\rm in}$ project onto the input codespace, and let $K_y$ denote any Kraus operator for the incoming error map followed by this cycle, associated with the recorded syndrome $y$. The shared environment is retained explicitly, and any additional Kraus label is suppressed.} 
Every 
\atchg{such Kraus operator, restricted to the input codespace,} is a sum of terms
\begin{equation}
 K_y\Pi_{\rm in}
 =\sum_\nu\bigl[z_\nu\widehat F(y)e_\nu\Pi_{\rm in}\bigr]\otimes B_\nu,
 \label{eq:record-form}
\end{equation}
where $\widehat F(y)$ is the same decoder correction in every term. \atchg{Each $B_\nu$ is an operator on the retained environment, and $\nu$ labels terms in the local operator expansion. The projector $\Pi_{\rm in}$ acts as the identity on that environment.}
For fixed constants $h_d,h_s,h_l$,
\begin{align}
 r(e_\nu)&\le r_{\rm in}+h_df,&
 |m_\nu|&\le h_sf,\nonumber\\
 r(z_\nu)&\le h_lf,&
 y&=\operatorname{syn}(e_\nu)+m_\nu.
 \label{eq:record-estimates}
\end{align}
The symbols $e_\nu,z_\nu$ denote Pauli operators and $m_\nu$ a binary syndrome error. The operators $B_\nu$ include environment and Kraus indices; they are amplitudes in an operator expansion, not probabilities.
\end{lemma}
\begin{proof}
Dilate \rev{the operation at every faulty location} and expand its local system part in a Pauli basis, retaining the environment coefficients. For a preparation or reset, include discarded systems in the environment. Refine faulty measurement instruments by their actual recorded outcome and expand the resulting local Kraus operators in matrix units. These rectangular operators can equivalently be written using local Pauli insertions and the ideal measurement bras. Propagation through the fixed Clifford circuit confines each new insertion to the forward light cone of its original location. This also treats interference between different operator terms.

For an incoming Pauli $E_0$, 
\atchg{let $V$ denote the ideal syndrome-extraction isometry, so that}
\begin{equation}
 V E_0\Pi_{\rm in}
 =E_0\Pi_{\rm in}\otimes|\operatorname{syn}(E_0)\rangle.
\end{equation}
Let $P$ be the propagated new data Pauli and $\Delta$ the resulting change to the ideal record in a particular surviving term. Its final extraction error is $e=PE_0$, whereas
\begin{equation}
 y=\operatorname{syn}(E_0)+\Delta
  =\operatorname{syn}(e)+\underbrace{\Delta+\operatorname{syn}(P)}_{m}.
 \label{eq:late-hook}
\end{equation}
A fault late in syndrome extraction can therefore change both the final data error and the discrepancy between its syndrome and the observed record. Both changes have size $O(f)$ because light cones and check incidence are bounded. There is no term proportional to $r_{\rm in}$ in the bound on $m$: ideal extraction reads the incoming error's own syndrome correctly.

Faults during decoder waits or an applied correction layer can be placed after the ideal Pauli correction. Their resulting data operator is $z$, with reduced weight $O(f)$. Every compatible term is corrected using the one deterministic function $\widehat F(y)$ of the actual record. Summing the propagated terms proves the claimed form without removing density-operator cross terms.
\end{proof}

\subsection{Keeping the residual error correctable}
The remaining task is to prevent errors from accumulating across cycles. Combining the decoder contraction with Lemma~\ref{m:record}, and absorbing the two CSS components into the constants, gives
\begin{equation}
 r_{\rm out}\le \lambda r_{\rm in}+Bf,\qquad
 \lambda=\alpha_j,\qquad B=\alpha_jh_d+\beta h_s+h_l.
 \label{eq:recurrence}
\end{equation}
This estimate applies within the decoder's stated error range. To keep every cycle in that range, choose $j$ first, so that $\lambda\le1/2$. \rev{The physical circuit depth and the bounds on fault propagation are now fixed.} Set $L=B/(1-\lambda)$ and choose $\eta>0$ sufficiently small that, for $\rho=L\eta$,
\begin{equation}
 2\rho<\delta,\qquad a(\rho+h_d\eta)+b h_s\eta<c.
 \label{eq:domain-choice}
\end{equation}
The left sides tend to zero with $\eta$, so this choice is possible.

If $r_{\rm in}\le\rho N$ and $f\le\eta N$, the second inequality in Eq.~\eqref{eq:domain-choice} puts every term in the decoder domain \emph{before} the contraction is invoked. Equation~\eqref{eq:recurrence} then gives $r_{\rm out}\le\lambda\rho N+B\eta N=\rho N$. Starting with an exact encoded state proves the bound after every cycle, provided that each preceding cycle also contains at most $\eta N$ faults. The stronger history-dependent estimate is
\begin{equation}
 r_t\le B\sum_{u=1}^t\lambda^{t-u}f_u.
 \label{eq:history}
\end{equation}
Each additional recovery suppresses the contribution of an earlier fault by a factor $\lambda$. Thus the decoder does not need to finish correcting the current error before the next cycle begins. Lemma~\ref{m:common} applies uniformly since $2\rho N<d_N$.

\subsection{Failure probability and complete cost}

Call a \rev{fault path} \emph{good} if every cycle has at most $\eta N$ faulty locations. The preceding argument shows that every \rev{good fault path} leaves only errors in $\mathsf E(\rho N)$, regardless of the reported syndromes. There are at most $c_0N$ physical locations in a cycle, with $c_0$ including the decoder waits. Put $m=\lfloor\eta N\rfloor+1$. If $m\le\lfloor c_0N\rfloor$, the local-stochastic condition gives
\begin{equation}
 \Pr(f>\eta N)\le \binom{\lfloor c_0N\rfloor}{m}p^m
 \le\left(\frac{e c_0Np}{m}\right)^m
 \le A e^{-\gamma N}
 \label{eq:tail}
\end{equation}
for sufficiently small fixed $p$, with constants $A,\gamma>0$. Otherwise the event is empty. This estimate is applied to the original fault locations. No conditioning on an observed syndrome is needed.

A union bound over the $S$ cycles gives $P_{\rm bad}\le AS e^{-\gamma N}$, even when faults in different cycles are correlated. On each \rev{good fault path}, Lemma~\ref{m:common} supplies the same recovery and restores the logical state, including its reference. Summing all measurement branches \rev{for that fixed fault path} gives a trace-preserving channel. Averaging over the input-independent fault distribution therefore yields
\begin{equation}
 \mathcal R\circ\mathcal P\circ\mathrm{Enc}
 =(1-P_{\rm bad})\mathrm{id}_K+P_{\rm bad}\mathcal B,
 \label{eq:channel-mixture}
\end{equation}
where $\mathcal B$ is a channel if $P_{\rm bad}>0$ and is arbitrary otherwise. The mixture is over \rev{classical fault paths}, not terms in the Pauli expansion. Consequently,
\begin{equation}
 \|\mathcal R\circ\mathcal P\circ\mathrm{Enc}-\mathrm{id}_K\|_\diamond
 \le 2AS e^{-\gamma N}.
 \label{eq:memory-error-bound}
\end{equation}

Choose an available block length above
\begin{equation}
 x=\max\left\{N_0,\frac K R,
                 \frac1\gamma\log\frac{2AS}{\varepsilon}\right\}.
 \label{eq:block-choice}
\end{equation}
Equation~\eqref{eq:tanner-density} allows $N\le c_Nx$. The block then has capacity for $K$ logical qubits and meets the target error. Each cycle uses $O(N)$ locations, $O(N)$ live qubits, and $O(N)$ classical work. This proves Theorem~\ref{m:memory}.
Each cycle uses $O(N)$ locations, $O(N)$ physical qubits, and $O(N)$ classical work. This proves Theorem~\ref{m:memory}.

\begin{corollary}[Optimal memory scaling under independent erasures]
\label{m:tight}
In the independent 
\atchg{erasure noise model}, assume that 
\atchg{after each elementary circuit time step, every physical qubit still retained by the protocol is independently erased with a fixed probability $p_{\rm er}>0$, with the erased locations reported and the qubits replaced}.
\atchg{When the resulting circuit-location noise has strength $p<p_*$ for the threshold $p_*$ of Theorem~\ref{m:memory},} the minimum total circuit cost obeys
\begin{equation}
 C_{\min}=\Theta\!\left(S\left(K+\log\frac S\varepsilon\right)\right).
 \label{eq:tight}
\end{equation}
\end{corollary}
\begin{proof}
The upper bound is Theorem~\ref{m:memory}. The storage-location lower bound of Ref.~\cite{BhartiHaugTanggara2026} applies even with ideal adaptive recovery. Total circuit cost is at least storage cost, and a constant-depth cycle changes duration by only a constant factor. This proves the matching lower bound. 
It is asserted for the stated 
\atchg{independent erasure noise model}, not for every distribution satisfying Eq.~\eqref{eq:ls}, which includes the zero-fault distribution.
\end{proof}

\subsection{Destructive logical readout}

\atchg{The encoded-output error criterion in Eq.~\eqref{eq:error}} does not require a final physical correction to the exact codespace.
If the desired output is classical, the block can instead be measured and decoded. The longer final decoding then takes place after the data have been measured, as in the readout procedure discussed in Section~3.3 of Ref.~\cite{GuEtAl2024}. Here we also track coherent fault operations and the conditional state of an external reference.

\begin{lemma}[Destructive Tanner readout]
\label{m:readout}
\atchg{Theorems~\ref{m:memory} and~\ref{m:clifford}} also allow a terminal measurement of every logical qubit of a block in the $Z$ basis, or every logical qubit in the $X$ basis.
\atchg{Choose the incoming residual bound $r(P)\le\rho N$ as in Eq.~\eqref{eq:domain-choice}, and a constant terminal fault fraction $\theta>0$ satisfying $a(\rho+h\theta)<c$, where $h$ bounds the number of physical outcome-bit errors per faulty terminal location as in the proof. Then} destructive physical measurement followed by the full Tanner decoder gives the ideal logical measurement instrument whenever the preceding \rev{fault path} is good and the terminal measurement layer has at most $\theta N$ faults.
\atchg{Including the terminal layer preserves the asymptotic diamond-norm error bounds~\eqref{eq:memory-error-bound} and~\eqref{eq:clifford-bound} for the intended logical measurement instrument, and the quantum-cost scalings~\eqref{eq:memory-result} and~\eqref{eq:clifford-cost}, with changed constants and, if necessary, a smaller noise threshold.}
\end{lemma}

\begin{proof}
Consider $Z$ readout and write a Pauli term of the incoming state using a stabilizer-equivalent representative of weight at most $\rho N$. Ideal physical $Z$ outcomes have support in $C_X=\ker H_Z$. Expand each faulty measurement instrument at its actual reported outcome. Let $y$ be the actual reported string after the known bit relabeling specified by the incoming Pauli frame; with physical corrections it is the raw reported string. This relabeling is common to every term and is included in the terminal classical processing. For every surviving term, $y=v+e$, where $v\in C_X$ and $|e|\le\rho N+h f$ for a fixed constant $h$ and $f$ terminal measurement faults. Stabilizer replacement changes only a scalar on the encoded input.

Run the full $X$-error decoder on $H_Zy=H_Ze$, with no syndrome noise, and call its correction $c(y)$. Shrink $\eta$ in Eq.~\eqref{eq:domain-choice}, if necessary, and choose a fixed terminal fault fraction $\theta>0$ such that
\begin{equation}
 a(\rho+h\theta)<c.
 \label{eq:readout-margin}
\end{equation}
For $f\le\theta N$, choose $j_{\rm fin}=O(\log N)$ such that $\alpha_{j_{\rm fin}}N<1$. Fact~\ref{m:decoder} with noiseless syndrome gives an integer reduced residual weight smaller than one, hence $e+c(y)\in\mathcal S_X$. Consequently
\begin{equation}
 d_Z(y)_i=z_i\cdot[y+c(y)]=z_i\cdot v
 \label{eq:tanner-readout}
\end{equation}
is the same ideal logical outcome for every term compatible with $y$.

More explicitly, \rev{the encoding in Eq.~\eqref{eq:css-encoding} obeys}
$\langle v_Z|\mathrm{Enc}_Q
 =2^{-\dim\mathcal S_X/2}\langle r_Z(v)_Z|$
for $v\in C_X$, where $r_Z(v)_i=z_i\cdot v$; the bra vanishes otherwise. A propagated Pauli contributes only a scalar phase and the bit shift $e$. \rev{Identify the correctable subspace with a logical subsystem and a syndrome subsystem. On these subsystems, every branch at record $y$ has the form}
$\langle d_Z(y)_Z|\otimes B_y$,
with environment and Kraus indices included in $B_y$. For a decoded result $r$, sum $B_y^\dagger B_y$ over its fine records and call the result $M_r$. Trace preservation gives
$\sum_r|r\rangle\langle r|\otimes M_r=I$,
so orthogonality implies $M_r=I$ for every $r$. This proves the full measurement instrument, including the conditional state of an arbitrary reference, without assigning probabilities to coherent Pauli terms.

The local-stochastic bound makes $f>\theta N$ exponentially unlikely by the same counting argument as Eq.~\eqref{eq:tail}. Adding this one terminal bad event changes only the constants in the total error bound. The physical measurement layer uses $O(N)$ locations. 
\atchg{Since all output data have been destructively measured,} the subsequent classical iterations add no quantum waits.
Exchange $X$ and $Z$ for the other basis.
\end{proof}

\rev{For specified logical representatives} $z_i$, evaluating all parities in Eq.~\eqref{eq:tanner-readout} directly takes $O(Nk_N)$ additional bit operations, or $O(NK)$ if only the $K$ stored coordinates are requested. These representatives can be dense. We therefore distinguish \rev{the $O(N)$ classical work of each recovery cycle while the quantum data are stored} from the final classical conversion to logical outcomes. Neither this conversion nor the final decoder requires quantum waits after all output blocks have been measured. \rev{Measuring selected logical qubits in different Pauli bases requires a separately specified physical circuit.}

For encoded quantum output, the theorem returns the block and its frame. A full final decoder run while quantum data remain stored would add $O(N\log N)$ waits, and faults during those waits would require continued protection. Such a step is not needed for \rev{the encoded-output error criterion in Eq.~\eqref{eq:error}}. \rev{Returning the quantum output to unencoded physical qubits exposes it to noise during the final decoding operations}~\cite{ChristandlFawziGoswami2024,ChristandlFawziGoswami2026}.

\section{Extension to logical operations}
\label{sec:clifford}
The memory can also support logical operations without changing the recovery depth. A concrete example is transversal CNOT between two identical Tanner blocks. \rev{Each block encodes $k_N$ logical qubits. Let $\mathbf a,\mathbf b\in\mathbb F_2^{k_N}$ label their logical computational basis states, with $|\mathbf a\rangle_L=\mathrm{Enc}_Q|\mathbf a\rangle$. Applying $N$ disjoint physical CNOT gates, one between each pair of corresponding physical qubits, gives
\begin{equation}
 |\mathbf a\rangle_L\otimes|\mathbf b\rangle_L
 \longmapsto
 |\mathbf a\rangle_L\otimes|\mathbf a+\mathbf b\rangle_L,
 \label{eq:transversal-cnot}
\end{equation}
where addition is componentwise modulo two. Thus the physical layer implements $k_N$ logical CNOT gates simultaneously, one on each corresponding logical pair.} An error can spread from one block to the other, but only by a bounded factor. A fixed number of additional decoder iterations absorbs this increase.

\subsection{Recovery between physical Clifford circuits}
The same reasoning applies whenever a logical Clifford operation has a constant-depth physical Clifford implementation with bounded-size gates. The implementation must map the input codespace to an output code with the same uniform decoder guarantees. Its \emph{error spread} $s$ is the maximum number of qubits in the image of a single-qubit Pauli under conjugation by the ideal circuit. Bounded depth and gate size make this spread constant.

\begin{theorem}[Recovery between logical Clifford layers]
\label{m:clifford}
Fix a constant number $b_0$ of length-$N$ Tanner blocks. 
Suppose every prescribed logical layer has a physical Clifford implementation of constant depth, bounded gate size, and $O(N)$ locations. 
The circuit maps its input product code to its declared output product code, implements the stated logical unitary in the fixed encodings, and has physical Pauli spread at most a fixed constant. 
\atchg{Assume in addition that every output code satisfies Fact~\ref{m:decoder} with constants independent of $N$, $T$, and the layer index.}
\atchg{Let $\mathcal P_T$ be the noisy physical channel for the complete sequence of $T$ gate-and-recovery layers, including all decoder waits and the output Pauli-frame register when used, and let $\mathcal U_t$ be the ideal logical unitary channel of layer $t$. The maps $\mathrm{Enc}$ and $\mathcal R$ are the fixed input encoding and output recovery-and-decoding maps defined analogously to Eq.~\eqref{eq:error} for the declared input and output codes.}
There is a threshold $p_*>0$, depending on these uniform bounds, such that at every fixed $p<p_*$ a recovery cycle after each of $T$ layers gives
\begin{equation}
 \|\mathcal R\circ\mathcal P_T\circ\mathrm{Enc}
       -\mathcal U_T\circ\cdots\circ\mathcal U_1\|_\diamond
 \le 2A T e^{-\gamma N},
 \label{eq:clifford-bound}
\end{equation}
for constants $A,\gamma>0$. If these implementations and \rev{the prescribed allocation of the $K$ logical input qubits to code blocks} are available uniformly along the code family, one may choose $N=O(K+1+\log(T/\varepsilon))$ and obtain
\begin{equation}
 C=O\!\left(T\left(K+1+\log\frac T\varepsilon\right)\right).
 \label{eq:clifford-cost}
\end{equation}
The schedule and logical layers are fixed. Their declared action includes all logical coordinates of the blocks. If only $K$ coordinates carry the task, the encoding fixes the remaining inputs and the output map discards the unused coordinates; the induced ideal channel must be the intended task. The input and output are encoded, and all decoder waits are included. The theorem assumes the physical gate implementations; it does not synthesize arbitrary logical Clifford circuits.
\end{theorem}

The full proof is in Appendix~\ref{app:clifford-proof}. I have made the existing reference in the following explanatory paragraph explicitly identify Theorem~\ref{m:clifford}.

The proof changes only the contraction estimate. If a gate increases the reduced error bound by a factor at most $g$, the coefficient multiplying the previous cycle's error becomes $g\alpha_j$ instead of $\alpha_j$. Since $\alpha_j$ decreases exponentially with the number $j$ of parallel iterations, one fixed choice of $j$ makes $g\alpha_j<1$. Each complete gate-and-recovery cycle still has $O(N)$ locations, so the same fault-count estimate and common recovery give the claimed error bound. The full proof of Theorem~\ref{m:clifford} is given in Appendix~\ref{app:clifford-proof}.

Besides transversal CNOT, a logical Pauli product has a physical Pauli representative that can be applied in one parallel layer once specified. \rev{Applying Hadamard to every physical qubit exchanges the $X$-type and $Z$-type stabilizer generators. In the output CSS code, choose logical $X$ and $Z$ representatives $z_i$ and $x_i$, respectively. With this choice of logical basis, the physical operation implements Hadamard on every logical qubit.} The decoder guarantees are unchanged after exchanging the two CSS components. \rev{The output codespace equals the input codespace only when the exchanged stabilizers define the same code; its logical action in an independently fixed basis must still be checked.} These are standard CSS-code operations~\cite{GottesmanBook2024}; the result here is that noisy recovery protects them with the stated cost.

\subsection{Selecting individual logical pairs}
Transversal CNOT acts on all corresponding logical pairs. Selecting only some pairs requires a different physical circuit: omitting physical CNOTs from the transversal layer does not, in general, preserve the code or select the desired logical qubits.

Appendix~\ref{s:cnots} characterizes all directed CNOT circuits between two identical CSS blocks that implement a prescribed logical CNOT matrix. The characterization gives an affine family of physical CNOT matrices. A matrix with bounded row and column weights yields a constant-depth circuit with linear cost and bounded error spread, and therefore satisfies Theorem~\ref{m:clifford}. The remaining requirement is to find such sparse representatives uniformly for the desired logical operations. We state this condition explicitly in Assumption~\ref{s:sparseassumption}; the algebraic characterization alone does not establish it. Thus transversal operations provide unconditional examples, while the same cost for selected logical pairs remains conditional on an additional circuit construction.

\section{Discussion}
\label{sec:discussion}
We have shown that an encoded quantum memory can attain optimal spacetime scaling with noisy recovery and efficient decoding. Storing $K$ logical qubits for $S$ cycles uses $O(S(K+\log(S/\varepsilon)))$ physical locations, including syndrome ancillas and all waits during classical processing. The upper bound holds for the stated class of coherent and correlated fault operations and matches the independent-erasure lower bound. Thus the noise introduced by recovery does not force an additional asymptotic memory cost.

The same block supplies logical capacity and reliability. Its redundancy grows with the logarithm of the storage duration and inverse error, and this contribution is shared across all stored qubits. The relative overhead is therefore constant once $K$ is at least of order $\log(S/\varepsilon)$. Separate surface-code patches pay a growing reliability cost for each logical qubit, while benefiting from local check interactions~\cite{DennisEtAl2002,FowlerEtAl2012}. Our result identifies what long-range connectivity and a good quantum LDPC code can achieve when the cost of noisy recovery is included.

The decoder does not have to remove every error before the next cycle begins. It only has to reduce the accumulated error sufficiently to leave room for new faults. The Tanner single-shot decoder already provides this contraction after a fixed number of iterations~\cite{GuEtAl2024}. Our circuit analysis shows that the contraction survives faults during extraction, correction, and decoder waits, with $O(N)$ classical work per cycle. Earlier syndromes can be discarded, and neither the number of iterations nor the syndrome history grows with storage time. Constant physical latency still requires enough parallel classical hardware and fixed-time communication along the decoder graph.

This reasoning also permits computation within the protected memory. A bounded-depth physical Clifford layer spreads each error by only a constant factor, which can be compensated by a fixed increase in the number of decoder iterations. Consequently, the cost remains $O(T(K+\log(T/\varepsilon)))$ for the supported depth-$T$ schedules. Transversal CNOT between identical blocks and specified logical Pauli layers give direct examples; whole-block Hadamard is covered with the corresponding change to the output code. These gates can be interleaved with recovery without introducing a growing delay. 

Selecting individual logical qubits is the next gate-construction problem. Our affine characterization lists every directed physical CNOT circuit with a given logical action, and bounded row and column weights suffice for constant-depth execution. The unresolved question is whether the required sparse matrices can be found for every desired logical selection in a suitable Tanner family. Code switching and logical measurement offer alternative routes to addressability~\cite{CohenEtAl2022,HuangEtAl2023,GolowichChangZhu2025,ZhengJiangXu2025}. Combining such methods with this memory would require compatible code interfaces and a bound on all preparation and waiting costs. \rev{Previously proposed constant-depth Clifford constructions based on encoded ancillas illustrate this issue~\cite{ZhengEtAl2020}. Even when the gates acting on the data have constant depth, preparing the ancillas and protecting the data while they wait can contribute to the total cost.}

The noise assumptions also delimit the result. Fault operations may be coherent and correlated, but their locations are sampled independently of the logical input and obey a local-stochastic bound. This suppresses histories with too many faults in a cycle. On all remaining histories, one common recovery corrects the entire linear space of possible error operators, including their coherent superpositions. The argument does not require the state to return to the codespace after each cycle. Extending the same optimal cost to noise models with small coherent perturbations at every location would require a different fault expansion; general-noise fault-tolerance results provide a relevant starting point~\cite{ChristandlFawziGoswami2025}.

A complete computational architecture must also prepare encoded states, admit unknown quantum inputs, and supply non-Clifford gates. Preparing a known code state, encoding an unknown state, and returning an unencoded quantum output have distinct requirements~\cite{ChristandlFawziGoswami2024,ChristandlFawziGoswami2026}. Establishing these interfaces for the same Tanner family would determine when the memory advantage survives in a larger computation. The higher time overheads in universal constructions~\cite{YamasakiKoashi2024,TamiyaKoashiYamasaki2024,NguyenPattison2025,HanLiu2026} account for capabilities beyond the protected layers considered here. A useful next step is to determine which of those capabilities can be added while preserving linear cost and constant duration per layer. %
Assessing practical performance will require explicit schedules for finite codes, estimates of decoder communication costs, and noise thresholds. These will determine the physical resources needed for a useful device beyond the asymptotic scaling established here.

\section*{Acknowledgements}

AT is supported by the CQT PhD scholarship, the Google PhD fellowship program, and the CQT Young Researcher Career Development Grant. Generative artificial intelligence tools were used to assist with ideation, language editing, organization, and the preparation of portions of the manuscript.

\begingroup
\small
\setlength{\bibsep}{1pt}
\bibliographystyle{unsrtnat}
\bibliography{references}
\endgroup

\clearpage
\appendix
\section{Proof of the Clifford extension}
\label{app:clifford-proof}
We give the modifications to the memory proof needed when a physical Clifford layer precedes each recovery cycle. The layer can spread existing errors, but by only a constant factor; the decoder compensates for this increase with a fixed number of iterations.

\begin{proof}[Proof of Theorem~\ref{m:clifford}]
For several blocks, extend $r(P)$ in Eq.~\eqref{eq:radius} by summing the reduced $X$ and $Z$ weights over all blocks. \rev{Every Pauli with $r(P)\le r$ still has a stabilizer-equivalent representative of physical weight at most $r$.} If $2r<d_N$, the proof of Lemma~\ref{m:common} applies to the product code.

Let $U$ be one of the given Clifford implementations. It maps input stabilizers to output stabilizers. Reducing a Pauli before conjugation, using the spread bound, and then separating the output $X$ and $Z$ components gives
\begin{equation}
 r_{\rm out}(UPU^\dagger)\le g r_{\rm in}(P),\qquad g=2s.
 \label{eq:spread}
\end{equation}
For a CNOT circuit, $X$ and $Z$ errors propagate separately and $g=s$ suffices. The circuit in Lemma~\ref{m:record} now starts with $U$. Its branch expansion has $e_\nu U\Pi_{\rm in}$ in place of $e_\nu\Pi_{\rm in}$ and satisfies $r(e_\nu)\le g r_{\rm in}+h_df$. The bounds on new syndrome errors and errors during waits are unchanged in form, with constants for the complete gate and recovery circuit.

The resulting recurrence is
\begin{equation}
 r_{\rm out}\le g\alpha_j r_{\rm in}
       +(\alpha_jh_d+\beta h_s+h_l)f.
 \label{eq:clifford-recurrence}
\end{equation}
Choose a fixed $j$ so that $\lambda=g\alpha_j\le1/2$, set $B=\alpha_jh_d+\beta h_s+h_l$ and $\rho=B\eta/(1-\lambda)$, and take $\eta>0$ sufficiently small that
\begin{equation}
 2\rho<\delta,\qquad a(g\rho+h_d\eta)+b h_s\eta<c.
\end{equation}
This keeps every Pauli term in the decoder's domain and preserves $r\le\rho N$ whenever a complete layer has at most $\eta N$ faults. The fault-count estimate~\eqref{eq:tail} applies with the new constant number of locations per data qubit. For sufficiently small fixed $p$, a union bound over $T$ layers gives $P_{\rm bad}\le AT e^{-\gamma N}$. The common recovery gives the intended logical unitary on every \rev{good fault path}, proving Eq.~\eqref{eq:clifford-bound}.

Under \rev{the assumptions on uniform gate implementations and allocation of logical qubits}, choose $N$ as in Eq.~\eqref{eq:block-choice}, with $S$ replaced by $T$ and with sufficient capacity for \rev{the prescribed allocation of logical qubits}. A fixed number of blocks changes the cost only by a constant. Each complete layer has $O(N)$ locations and physical qubits, which proves Eq.~\eqref{eq:clifford-cost}.
\end{proof}

\section{Physical circuits for selected logical CNOT gates}
\label{s:cnots}

This appendix supplies the algebraic characterization used in Section~\ref{sec:clifford}. We consider circuits in which all controls lie in one CSS block and all targets in a second identical block. These CNOT gates commute, so a binary matrix specifies the circuit. We first characterize the matrices with a prescribed logical action, then state the additional sparsity condition needed for constant-depth implementation. The characterization is unconditional and applies to any CSS code; its efficient application to Tanner blocks remains conditional on that sparsity condition. General logical Clifford synthesis can be formulated through binary symplectic constraints~\cite{RengaswamyEtAl2020}. Here we restrict to directed CNOT circuits between two CSS blocks, for which an explicit affine parametrization exposes the row and column weights that control physical depth.

\subsection{Physical circuits for a prescribed logical action}

All vector spaces are over $\mathbb F_2$. Write $S_X=\mathcal S_X$ and $S_Z=\mathcal S_Z$ for the stabilizer spaces of Section~\ref{sec:model}, and define
\begin{equation}
 S_X=\operatorname{row}H_X,\qquad
 S_Z=\operatorname{row}H_Z,\qquad
 C_X=\ker H_Z,\qquad C_Z=\ker H_X.
 \label{s:cssspaces}
\end{equation}
We use column vectors for physical Pauli supports.  The CSS condition is $S_X\subseteq C_X$.  Write $r_X=\dim S_X$, $r_Z=\dim S_Z$, and $k=N-r_X-r_Z$.  Choose full-column-rank matrices $G_X$ and $G_Z$ with images $S_X$ and $S_Z$.  These are algebraic choices used in the proof; the measured check lists $H_X,H_Z$ are not replaced by these bases.

Fix paired logical representatives $L_X,L_Z\in\mathbb F_2^{N\times k}$ satisfying
\begin{equation}
 C_X=S_X\oplus\operatorname{im}L_X,\qquad
 C_Z=S_Z\oplus\operatorname{im}L_Z,\qquad
 L_Z^TL_X=I_k.
 \label{s:paired}
\end{equation}
The columns of $L_X$ and $L_Z$ specify the paired logical Pauli operators. The basis is fixed throughout a computation. For a binary $N\times N$ matrix $A$, define the physical unitary
\begin{equation}
 U_A|u\rangle|v\rangle=|u\rangle|v+Au\rangle.
 \label{s:shear}
\end{equation}
It is the product of the CNOT gates from qubit $j$ in the first block to qubit $i$ in the second block for which $A_{ij}=1$.  These gates commute.  \rev{A binary matrix $M\in\mathbb F_2^{k\times k}$ specifies the logical action $|x\rangle_L|y\rangle_L\mapsto|x\rangle_L|y+Mx\rangle_L$, where $x,y\in\mathbb F_2^k$ label all $k$ logical qubits in each block. To select CNOT gates between corresponding logical pairs, take $M=D_m=\operatorname{diag}(m)$ for a binary vector $m\in\mathbb F_2^k$: apply the logical CNOT on pair $i$ if $m_i=1$, and act as the identity on that pair if $m_i=0$.}

\begin{theorem}[Exact CNOT representatives]
\label{s:affine}
\rev{For a CSS code with phase-$+1$ $X$-type and $Z$-type stabilizer generators} and the paired basis~\eqref{s:paired}, $U_A$ preserves the two-block code and has logical action $M$ if and only if
\begin{align}
 AS_X&\subseteq S_X,& A^TS_Z&\subseteq S_Z,
 \label{s:preserve}\\
 AL_X-L_XM&\in S_X^{\oplus k}.
 \label{s:logical}
\end{align}
Equivalently, the complete set of such matrices is
\begin{equation}
 A=L_XML_Z^T+G_XR+BG_Z^T,
 \quad
 R\in\mathbb F_2^{r_X\times N},\quad
 B\in\mathbb F_2^{N\times r_Z}.
 \label{s:affineformula}
\end{equation}
This affine space has dimension $N(r_X+r_Z)-r_Xr_Z$.  Its induced action on logical $Z$ operators is fixed by
\begin{equation}
 A^TL_Z-L_ZM^T\in S_Z^{\oplus k}.
 \label{s:dualaction}
\end{equation}
\end{theorem}

\begin{proof}
Pauli conjugation by $U_A$ gives
\begin{align}
 X(u)\otimes I&\longmapsto X(u)\otimes X(Au),&
 I\otimes X(v)&\longmapsto I\otimes X(v),\nonumber\\
 Z(u)\otimes I&\longmapsto Z(u)\otimes I,&
 I\otimes Z(v)&\longmapsto Z(A^Tv)\otimes Z(v).
 \label{s:pauliaction}
\end{align}
Consequently, code preservation is equivalent to~\eqref{s:preserve}.  The logical action of a control-block $X$ operator is $M$ exactly when~\eqref{s:logical} holds.  In fact, the first condition in~\eqref{s:preserve} and~\eqref{s:logical} already imply $AC_X\subseteq C_X$, and hence the second condition in~\eqref{s:preserve}.

These conditions determine the unitary without an undetected logical phase.  \rev{With these stabilizer generators, the code has the orthonormal basis}
\begin{equation}
 |x\rangle_L=|S_X|^{-1/2}\sum_{s\in S_X}|s+L_Xx\rangle.
 \label{s:cosetbasis}
\end{equation}
For each fixed $s\in S_X$, the target-block sum over $t\in S_X$ is translated by $As+(AL_X-L_XM)x\in S_X$.  Therefore $U_A$ sends the two coset sums to the coset sums labeled by $x,y+Mx$, with no phase.  This also proves the desired action on \rev{arbitrary superpositions and states entangled with an external reference}.

\rev{The particular solution $A_0=L_XML_Z^T$} satisfies all the conditions.  Every term $G_XR+BG_Z^T$ maps $C_X$ into $S_X$, so every matrix in~\eqref{s:affineformula} is a solution.  Conversely, let $D=A-A_0$ be the difference of two solutions.  Then $D(C_X)\subseteq S_X$.  Since $G_X$ is an isomorphism onto $S_X$, the restriction $D|_{C_X}$ has a linear lift through $G_X$.  Extend this lift to $R:\mathbb F_2^N\to\mathbb F_2^{r_X}$.  The map $D-G_XR$ vanishes on $C_X=\ker G_Z^T$.  Because $G_Z^T$ is surjective, it factors as $BG_Z^T$ for some $B$.

The parameter map $(R,B)\mapsto G_XR+BG_Z^T$ has kernel precisely
\begin{equation}
 (R,B)=(JG_Z^T,G_XJ),\qquad
 J\in\mathbb F_2^{r_X\times r_Z}.
 \label{s:parameterkernel}
\end{equation}
Indeed, its vanishing first implies that $R$ vanishes on $C_X$, and then surjectivity of $G_Z^T$ determines $B$.  Subtracting this kernel dimension gives the stated affine dimension.

Finally, each column of $A^TL_Z-L_ZM^T$ pairs to zero with $S_X$ by~\eqref{s:preserve}, and pairs to zero with $\operatorname{im}L_X$ by~\eqref{s:logical} and~\eqref{s:paired}.  It therefore lies in $C_X^\perp=S_Z$, proving~\eqref{s:dualaction}.
\end{proof}

\subsection{When the circuit has constant depth}
The affine description determines which physical CNOTs give the right logical gate. Their scheduling cost depends on how many gates share each physical qubit. Bounded row and column weights give precisely the required bound.

\begin{lemma}[Cost of a sparse representative]
\label{s:sparsecost}
Suppose the row and column weights of $A$ are at most a constant $\Delta_A\ge1$.  Then $U_A$ uses at most $\Delta_A N$ CNOT gates, \rev{no ancillary qubits}, and at most $2\Delta_A-1$ layers.  \rev{It maps every single-qubit Pauli error to a Pauli error of weight at most $\Delta_A+1$.}  Given the sparse edge list, a schedule with these properties can be found in $O(N)$ word operations.
\end{lemma}

\begin{proof}
Make a bipartite graph with one vertex for each control and target qubit and one edge for each nonzero matrix entry.  Its maximum degree is $\Delta_A$.  Greedy edge coloring uses at most $2\Delta_A-1$ colors because an edge meets at most $2\Delta_A-2$ other edges.  Each color is a valid gate layer.  For constant degree, maintaining the used colors and assigning a color to each edge costs $O(N)$ word operations.  Equation~\eqref{s:pauliaction} bounds the spread of every single-qubit $X$, $Z$, or $Y$ by $\Delta_A+1$.
\end{proof}

Theorem~\ref{s:affine} does not supply a sparse matrix in its affine space.  \rev{The particular solution $L_XML_Z^T$ can be dense}, as can both logical-representative matrices.  Gaussian elimination gives a description of the space but does not prove that \rev{any prescribed selection of logical pairs} has \rev{a representative with bounded row and column weights}.

\begin{assumption}[\rev{Sparse circuits for selected logical CNOT gates}]
\label{s:sparseassumption}
A good Tanner-code family with the block-length property of Section~\ref{sec:model} satisfying the hypotheses of Theorem~\ref{m:clifford} has one fixed paired logical basis for each block length.  There is a constant $\Delta_A$ such that \rev{every selection vector $m$} has a representative~\eqref{s:affineformula} with row and column weights at most $\Delta_A$.  A classical algorithm outputs its sparse edge list in $O(N)$ word operations \rev{per selection vector}.  Any size-dependent setup is included in a reusable $O(N)$-work setup, and the algorithm does not require \rev{a table of all selections whose construction and storage are excluded from the cost} or an uncharged dense change of logical basis.
\end{assumption}

\begin{corollary}[Conditional cost of selected CNOT schedules]
\label{s:cnotcorollary}
Under Assumption~\ref{s:sparseassumption}, and at every fixed $p$ below the threshold of Theorem~\ref{m:clifford} for the resulting circuits, a prescribed depth-$T$ schedule of selected CNOT layers on a fixed number of encoded blocks of the chosen size can be executed \rev{with encoded quantum output}, logical diamond-norm error at most $\varepsilon$, and
\begin{align}
 C_{\rm FT}&=O\!\left(T[K+\log(T/\varepsilon)]\right),\nonumber\\
 Q_{\rm peak}&=O\!\left(K+\log(T/\varepsilon)\right).
 \label{s:cnotcost}
\end{align}
Setup, compilation, and scheduling use $O(TN)$ word operations, \rev{while error correction during the computation uses} $O(TN)$ \rev{Boolean gates with bounded fan-in}.  Compilation is performed \rev{before the encoded input is supplied}.  \rev{Each logical layer and its following recovery take constant physical time.}  Explicit storage of all compiled edge lists uses $O(TN)$ address words, or $O(TN\log N)$ bits.
\end{corollary}

\begin{proof}
Theorem~\ref{s:affine} proves the exact logical action, and Lemma~\ref{s:sparsecost} supplies a constant-depth, linear-work physical layer of uniformly bounded spread.  Apply Theorem~\ref{m:clifford}.  Choose an available block length above the capacity and reliability requirements, $N\ge c_1K+c_2\log(T/\varepsilon)$, with constants sufficient for that theorem.  The block-length property bounds the excess length by a constant factor.  Compiling and scheduling \rev{all selections of logical pairs} takes $O(TN)$ work.  \rev{Every wait while quantum data remain stored} is included in the execution cost of Theorem~\ref{m:clifford}.
\end{proof}

The assumption concerns corresponding logical pairs in a fixed basis.  The corollary does not construct arbitrary logical matchings, \rev{Hadamard gates on individually selected logical qubits}, or non-Clifford gates.

\end{document}